\documentclass{amsart}
\usepackage{latexsym, amsthm, amsmath, bbm}
\usepackage{amssymb}
\usepackage{amsfonts}
\usepackage{amstext}
\usepackage{graphicx}
\usepackage{multicol}
\usepackage{subfigure}
\usepackage{textcomp}
\usepackage{enumerate}
\usepackage{verbatim}
\usepackage{comment}
\usepackage{xcolor}
\usepackage{float}
\usepackage{pdfpages}
\usepackage{breqn}
\usepackage{epstopdf}
\usepackage{pgfplots}
\usepackage{tabu}
\usepackage{youngtab}

\newtheorem{Theorem}{Theorem}

\newtheorem{Proposition}{Proposition}
\newtheorem{definition}{Definition}
\newtheorem{Corollary}{Corollary}

\theoremstyle{remark}

\theoremstyle{remark}

\newcommand{\beq}{\begin{equation}}
\newcommand{\beql}[1]{\begin{equation}\label{#1}}
\newcommand{\eeq}{\end{equation}}

\theoremstyle{remark}
\newtheorem{example}{Example}

\newcommand{\C}{{\mathbb{C}}}
\newcommand{\Z}{{\mathbb{Z}}}
\newcommand{\Q}{{\mathbb{Q}}}

\newcommand {\N}{{\mathbb{N}}}

\DeclareMathOperator{\sgn}{sgn}

\DeclareMathOperator{\CT}{CT}

\begin{document}

\title{Inverse Participation Ratios in the XX spin chain}

\author{Emmanuel Tsukerman}
\date{\today}
\address{Department of Mathematics, University of California,
Berkeley, CA 94720-3840}
\email{e.tsukerman@berkeley.edu}

\begin{abstract}
We continue the study of the Inverse Participation Ratios (IPRs) of the XXZ Heisenberg spin chain initiated by Misguich, Pasquier and Luck (2016) by focusing on the case of the XX Heisenberg Spin Chain. For the ground state, Misguich et al. note that calculating the IPR is equivalent to Dyson's constant term ex-conjecture. We express the IPRs of excited states as an apparently new ``discrete" Hall inner product. We analyze this inner product using the theory of symmetric functions (Jack polynomials, Schur polynomials, the standard Hall inner product and $\omega_{q,t}$) to determine some exact expressions and asymptotics for IPRs. We show that IPRs can be indexed by partitions, and asymptotically the IPR of a partition is equal to that of the conjugate partition. We relate the IPRs to two other models from physics, namely, the circular symplectic ensemble of Dyson and the Dyson-Gaudin two-dimensional Coulomb lattice gas. Finally, we provide a description of the IPRs in terms of a signed count of diagonals of permutohedra.  
\end{abstract}

\maketitle

\section{Introduction}

Recently in \cite{misguich2016inverse}, Misguich, Pasquier and Luck began a numerical study of the so-called Inverse Participation Ratios (IPRs) in the spin-1/2 XXZ  chains. IPRs give a measure of localization of the system with respect to a preferential basis. Roughly speaking, large IPRs imply that the system will be found in a single state, whereas small IPRs imply that the system is equally likely to be found in any available state. The interest in the question stems from the fact that integrable systems generally fail to reach thermal states. Misguich et al. ask in particular about the IPRs in the special case of the XX spin chain  with respect to the Ising basis.

We let $L$ denote the length of the spin chain and restrict attention to the space spanned by Ising bases of $M$ down spins. To analyze the IPRs, we will proceed as follows. First, we will employ the Bethe ansatz to obtain an expression for the eigenstates in terms of the Ising basis. We will see then that the IPRs are naturally labeled by partitions, and are equal to 
\beq
\mathbf{t}_{\lambda}=(M!L^{2M})^{-1}  \sum_{\theta_1,\ldots,\theta_M \in \{0,\frac{2 \pi}{L},\ldots,\frac{2 \pi(L-1)}{L}\}} |s_\lambda(e^{i \theta_1},\ldots,e^{i \theta_M})|^4 |\prod_{i<j}|e^{i \theta_i}-e^{i \theta_j}|^4 , 
\eeq
where $s_\lambda$ is the Schur polynomial (a specialization of the Jack polynomials). This will naturally lead us to the theory of symmetric functions. We will analyze these quantities in complete mathematical rigor (section \ref{schurjack} and on).

Before employing techniques from symmetric function theory, we provide several physical interpretations of the IPRs and discuss the ground state. The calculation of the ground state IPR ($s_\lambda=1$) is equivalent to Dyson's constant term ex-conjecture, an important motivation for much of the development of the theory of Macdonald polynomials. For excited states, we can interpret the IPR as the expected value of $|s_\lambda|^4$ over the Dyson-Gaudin two-dimensional Coulomb lattice gas. In the limit $L \rightarrow \infty$ the gas is no longer restricted to a lattice. In this situation, the IPR can be interpreted as the expected value of $|s_\lambda|^4$ over the circular symplectic ensemble (CSE) of Dyson. The CSE is a modification of the Gaussian symplectic ensemble (GSE), a key matrix ensemble in random matrix theory.

Next we study the IPRs using the tools of symmetric function theory, namely, Jack polynomials, Hall inner product and $\omega_{q,t}$. We interpret the IPRs in terms of a ``discrete" Hall inner product and show that assuming, roughly, $L>2M$ the discrete Hall inner product and the traditional one coincide. This allows us to use the orthogonality of the Jack polynomials to evaluate the IPRs in terms of the transition coefficients from Schur polynomials to Jack polynomials. When these are known, we are able to determine exact expressions for the IPRs and their asymptotics. Finally, we show that curiously, if $\lambda'$ is the conjugate partition of $\lambda$, then 
\beq
\mathbf{t}_\lambda=\mathbf{t}_{\lambda'}
\eeq
assuming, roughly, that $L>2M \gg 1$.

\tableofcontents

\section{Inverse Participation Ratios}

In this section we provide the basic definition of the Inverse Participation Ratios (\cite{edwards1972numerical}, \cite{misguich2016inverse}).

\begin{definition}
Let the normalized eigenvectors of a hamiltonian $\mathcal{H}$ be $\{|\psi_i \rangle\}_{i=1,\ldots,D}$ and assume they are non-degenerate. The inverse participation ratio (IPR) of an eigenstate $|\psi_k \rangle$ in a preferential basis $\{|a_i\rangle\}_{i=1,\ldots,D}$ is
\beq
\mathbf{t}_k:=\sum_{i=1}^D | \langle a_i | \psi_k \rangle|^4.
\eeq
\end{definition}

The maximum value of this quantity is reached when an eigenstate coincides with a single basis state, in which case $t_{\text{max}}=1$. The minimum value is reached for eigenstates which are uniform superpositions of all the basis states, with the same modulus $|\langle a_i | \psi_k \rangle|=\frac{1}{\sqrt{D}}$. This maximally delocalized limit gives $t_{\text{min}}=\frac{1}{D}$.

We note that if
\beq
\tilde{\psi}_k=\sum_i c_i |a_i \rangle,
\eeq
is an unnormalized multiple of $\psi_k$, then
\beq
|\langle a_i | \psi_k \rangle|^4=\frac{|c_i|^4}{(\sum_j |c_j|^2)^2}
\eeq
and
\beql{IPR}
\mathbf{t}_k=\frac{\sum_i |c_i|^4}{(\sum_j|c_j|^2)^2}.
\eeq

\section{Coordinate Bethe ansatz for the periodic XX Heisenberg model}

We consider the XX spin chain (XXZ with anisotropy parameter $\Delta=0$) with periodic boundary conditions and $L$ sites. The Hamiltonian is given by
\beq
\mathcal{H}=\sum_{i=1}^{L} S_i^x S_{i+1}^x +S_i^y S_{i+1}^y. 
\eeq

We will single out the preferential basis called the Ising configuration for the IPRs. These are the eigenstates of all $S_i^z$:
\beq
| \uparrow \uparrow \downarrow \cdots \rangle , | \uparrow \downarrow \uparrow \cdots \rangle, \ldots
\eeq  
We restrict attention to the subspace spanned by Ising configurations having $M$ down spins. This space has dimension $\binom{L}{M}$. 
The coordinate Bethe ansatz tells us that the wavenumbers $k_j,j=1,\ldots,M$, and Bethe roots $I_j,j=1,\ldots,M$, satisfy
\beq
Lk_j =2 \pi I_j, \quad j=1,2,\ldots,M.
\eeq
For simplicity, we will be assuming that $M$ is odd, so that the $I_j$ are in $\{0,1,\ldots,L-1\} \mod L$ \cite{vsamaj2013introduction}. All statements can be adapted to the case of even $M$, in which case the $I_j$ are half integers.
To avoid the nullity of the wavefunction, the wavenumbers must be distinct.

We use the shorthand notation 
\beq
|\mathbf{x} \rangle = |x_1 \cdots x_M \rangle, \quad x_1 < x_2 < \cdots < x_M
\eeq
for the Ising basis with down spins at $x_1,\ldots,x_M$ and
\beq
\mathbf{k}=(k_1,\ldots,k_M).
\eeq
The eigenvectors of $\mathcal{H}$ are now given by
\beq
\psi_{\mathbf{k}}=\sum_{\mathbf{x}} c(\mathbf{x}) |\mathbf{x} \rangle=\sum_{\mathbf{x}} \sum_{ \pi \in S_M} \sgn(\pi) e^{i \pi(\mathbf{k}) \cdot \mathbf{x}} |\mathbf{x} \rangle=\sum_{\mathbf{x}} \det(e^{ik_a x_b})_{a,b}|\mathbf{x} \rangle.
\eeq
By \eqref{IPR}, the IPRs are equal to
\beql{IPRforxx}
\mathbf{t}_k=\frac{\sum_{\mathbf{x}} |\det(e^{ik_a x_b})_{a,b}|^4 }{(\sum_{\mathbf{x}} |\det(e^{ik_a x_b})_{a,b}|^2 )^2}.
\eeq
The denominator of \eqref{IPRforxx} is relatively easy to evaluate. We have
\beq
|c(\mathbf{x})|^2=  c(\mathbf{x}) \overline{c}(\mathbf{x})= (\sum_{ P \in S_M} \sgn(P) e^{i P(\mathbf{k}) \cdot \mathbf{x}} ) (\sum_{ Q \in S_M} \sgn(Q) e^{-i Q(\mathbf{k}) \cdot \mathbf{x}} )=
\eeq
\beq
=\sum_{P,Q \in S_M} \sgn(P) \sgn(Q) e^{i(P(\mathbf{k})- Q(\mathbf{k})) \cdot \mathbf{x}}. 
\eeq
We now use the fact that the determinant is zero whenever any two of $x_i$ are equal. This allows us to remove the restriction of the $x_i$ being distinct. In addition, we must divide by $M!$ to account for the order. We also note that
\beq
\sum_{x_1,\ldots,x_M=1}^{L} \prod_{j=1}^M e^{ix_j m_j}=\prod_{j=1}^M\sum_{x_j=1}^{L}  e^{ix_j m_j}.
\eeq
We have
\beq
\sum_{\mathbf{x}}|c(\mathbf{x})|^2=\frac{1}{M!}\sum_{x_1,\ldots,x_M=1}^{L}|c(\mathbf{x})|^2=\frac{1}{M!}\sum_{x_1,\ldots,x_M=1}^{L}\sum_{P,Q \in S_M} \sgn(P) \sgn(Q) \prod_{j=1}^M  e^{i x_j(k_{P(j)}- k_{Q(j)})}
\eeq
\beq
=\frac{1}{M!}\sum_{P,Q \in S_M} \sgn(P) \sgn(Q) \sum_{x_1,\ldots,x_M=1}^{L} \prod_{j=1}^M  e^{i x_j(k_{P(j)}- k_{Q(j)})}=
\eeq
\beq
=\frac{1}{M!}\sum_{P,Q \in S_M} \sgn(P) \sgn(Q)\prod_{j=1}^M  \sum_{x_j=1}^{L}  e^{i x_j(k_{P(j)}- k_{Q(j)})}.
\eeq

Now $k_j=\frac{2 \pi I_j}{L}$ and $0<|I_j-I_i|<L$. Consequently,
\beq
\sum_{x_j=1}^L e^{i x_j(k_{P(j)}- k_{Q(j)})}=L\delta(k_{P(j)}=k_{Q(j)}).
\eeq
Substituting,
\beq
\sum_{\mathbf{x}}|c(\mathbf{x})|^2=L^M\frac{1}{M!}\sum_{P,Q \in S_M} \sgn(P) \sgn(Q)\prod_{j=1}^M  \delta(k_{P(j)}=k_{Q(j)})=L^M.
\eeq

Therefore
\beq
\mathbf{t}_k=\frac{\sum_{\mathbf{x}} |\det(e^{ik_a x_b})_{a,b}|^4 }{L^{2M}}.
\eeq

Next, set $\theta_j:=\frac{2 \pi}{L}x_j$, so that 
\beq
e^{ik_a x_b}=e^{i I_a \theta_b}.
\eeq

We set
\beq
c(\boldsymbol{\theta})=\det(e^{i I_a \theta_b})_{a,b}
\eeq
and introduce the notation $\sum_{\boldsymbol{\theta}}$ to denote summation over distinct values of $\theta_1,\ldots,\theta_M$ in $\frac{2 \pi}{L}\{0,\ldots,L-1\}$.

\begin{example}
\underline{IPR of the ground state}. For the ground state, 
\beq
\{I_1,I_2,\ldots,I_M\}=\{\frac{M-1}{2},\frac{M-1}{2}-1,\ldots,-\frac{M-1}{2}\}.
\eeq
Therefore
\beql{cdet}
c(\boldsymbol{\theta})=\det\left(\begin{array}{cccccc}
e^{i(-\frac{M-1}{2})\theta_1} &  e^{i(-\frac{M-1}{2})\theta_2} & \cdots & e^{i(-\frac{M-1}{2})\theta_M} \\
e^{i(-\frac{M-1}{2}+1)\theta_1} &  e^{i(-\frac{M-1}{2}+1)\theta_2} & \cdots & e^{i(-\frac{M-1}{2}+1)\theta_M} \\
\vdots & \vdots & \cdots & \vdots \\
e^{i(\frac{M-1}{2})\theta_1} &  e^{i(\frac{M-1}{2})\theta_2} & \cdots & e^{i(\frac{M-1}{2})\theta_M}
 \end{array}\right).
\eeq
We notice that \eqref{cdet} is equal to
\beq
e^{i(-\frac{M-1}{2})(\theta_1+\ldots+\theta_M)} \det\left(\begin{array}{cccccc}
1 & 1 & \cdots & 1 \\
e^{i\theta_1} &  e^{i\theta_2} & \cdots & e^{i\theta_M} \\
\vdots & \vdots & \cdots & \vdots \\
e^{i(M-1)\theta_1} &  e^{i(M-1)\theta_2} & \cdots & e^{i(M-1)\theta_M}
 \end{array}\right),
\eeq
the determinant being a Vandermonde determinant: 
\beq
\det\left(\begin{array}{cccccc}
1 & 1 & \cdots & 1 \\
e^{i\theta_1} &  e^{i\theta_2} & \cdots & e^{i\theta_M} \\
\vdots & \vdots & \cdots & \vdots \\
e^{i(M-1)\theta_1} &  e^{i(M-1)\theta_2} & \cdots & e^{i(M-1)\theta_M}
 \end{array}\right)=\prod_{j<k} (e^{i \theta_k}-e^{i \theta_j}).
\eeq
Therefore the IPR of the ground state is equal to
\beq
\mathbf{t}_0=\frac{1}{L^{2M}} \sum_{\boldsymbol{\theta}} \prod_{j<k} |e^{i \theta_k}-e^{i \theta_j}|^4.
\eeq
We will complete the evaluation in \eqref{ground}.
\end{example}

Notice that $|\det(e^{i I_a \theta_b})_{a,b}|$ is invariant under shifts of the Bethe roots $I$, since 
\beq
\det(e^{i (I_a+n) \theta_b})_{a,b}=e^{in( \theta_1+\ldots+\theta_M)}\det(e^{i I_a \theta_b})_{a,b}.
\eeq
Therefore IPRs corresponding to two excited states whose Bethe roots are related by a shift will be equal.

 Assume without loss of generality that $I_1 > I_2 > \ldots > I_M$. Set 
\beq 
 \lambda_j=I_j-\frac{M+1}{2}+j
\eeq 
   so that $\lambda_1 \geq \lambda_2 \geq \ldots \geq \lambda_M$ yields a partition $\lambda=[\lambda_1,\ldots,\lambda_M]$. In this notation, 
\beq
|\det(e^{i I_a \theta_b})_{a,b}|=|s_\lambda(e^{i \theta_1},\ldots,e^{i \theta_M}) V(e^{i \theta_1},\ldots,e^{i \theta_M})|,
\eeq
where $s$ is the Schur polynomial (discussed in Section \ref{schurjack}) and $V(x_1,\ldots,x_M)$ is the Vandermonde determinant:
\beq
V(x_1,\ldots,x_M)=\prod_{i<j} (x_j-x_i).
\eeq
We then have
\beql{IPR2}
\mathbf{t}_{\lambda}=\frac{\sum_{\boldsymbol{\theta}} |s_\lambda(e^{i \theta_1},\ldots,e^{i \theta_M})|^4 |\prod_{i<j}|e^{i \theta_i}-e^{i \theta_j}|^4 }{L^{2M}},
\eeq
where we have introduced a labeling by partitions $\lambda$. 

\section{Connection between the IPRs of XX, the circular symplectic ensemble and the Dyson-Gaudin Coulomb gas}

We show in this section how to interpret the IPRs of XX in terms of expectation values of quantities for two-dimensional Coulomb gas on a one-dimensional lattice, or equivalently, as expectation values over a discretization of the circular symplectic ensemble.

In studying his famous random matrix ensembles, Dyson \cite{dyson1962statistical} developed a physical model for the eigenvalues of the matrices of the circular ensemble consisting of charges distributed on a unit circle in two-dimensions and experiencing Coulomb forces. In \cite{gaudin1973gaz}, Gaudin studied a discrete formulation of the problem, in which the gas particles lie on lattice sites of the circle. This formulation corresponds to the finite $L$ situation of the IPRs above, whereas Dyson's to the limit as $L\rightarrow \infty$. We will now explain the details.

\subsection{Circular symplectic ensemble}

The circular symplectic ensemble (CSE) is the space of self dual unitary quaternion matrices with a probability measure defined as follows. Each element is a unitary matrix, so it has eigenvalues on the unit circle $e^{i \theta_1},\ldots,e^{i \theta_M}$. The probability density function for the phases in the circular symplectic ensemble is given by
\beq
p(\theta_1,\ldots,\theta_M)=\frac{1}{Z_{M,4}} \prod_{1 \leq i < j \leq M}|e^{i \theta_i}-e^{i \theta_j}|^4.
\eeq
The normalization constant is given by
\beq
Z_{M,4}=(2 \pi)^M \frac{(2M)!}{2^M}.
\eeq
\subsection{Dyson-Gaudin Coulomb gas}
 
The positions which a unit charge can occupy on the circumference of a unit circle are restricted to $L$ equidistant points $\exp(i \theta_j)$, $\theta_j=2 \pi j/L$, $1 \leq j \leq L$. One considers three distinguished values of the inverse temperature $\beta=1,2,4$\footnote{These correspond respectively to the orthogonal, unitary and symplectic circular ensembles, which are closely related to Dyson's threefold way.} The joint probability density for $M$ unit charges to occupy positions $j_1,\ldots,j_M$ is given by
\beq
P_\beta(j_1,\ldots,j_M)=C^{-1}_{LM\beta} L^{-M} \exp(-\beta W).
\eeq
Here $W$ is the potential energy, calculated as follows. If we place point unit charges at angles $\theta_1,\ldots,\theta_M$ on a circle in two dimensions of radius $1$, then the potential energy is equal to
\beq
W=-\sum_{1 \leq j < k \leq M} \log |e^{i \theta_k }-e^{i \theta_j}|, \quad \theta_l=2 \pi j_l/L.
\eeq
Note in particular that
\beq
\exp(-\beta W)=\prod_{1 \leq j < k \leq M} |e^{i \theta_k}-e^{j \theta_j}|^\beta. 
\eeq

The expected value of a quantity $f(e^{i \theta_1},\ldots,e^{i \theta_M})$ is given by
\beq
\mathbb{E}_\beta (f)=\sum_{\boldsymbol{\theta}} f(e^{i \theta_1},\ldots,e^{i \theta_M}) P_\beta(j_1,\ldots,j_M)= C^{-1}_{LM\beta} L^{-M}\sum_{\boldsymbol{\theta}} f(e^{i \theta_1},\ldots,e^{i \theta_M})\prod_{1 \leq j < k \leq M} |e^{i \theta_k}-e^{i \theta_j}|^\beta .
\eeq
 In \cite{gaudin1973gaz}, Gaudin  calculates the partition function of the discrete Coulomb gas. In particular, he determines the normalization constant for $\beta=4$: 
\beq
\sum_{\boldsymbol{\theta}} \prod_{1 \leq j < k \leq M}|e^{i \theta_k}-e^{i \theta_j}|^4=\frac{(2M)!L^M}{2^M M!}
\eeq
Hence
\beq
\mathbb{E}_{\beta=4}(f)= \frac{2^M M!}{(2M)! L^M}\sum_{\boldsymbol{\theta}} f(e^{i \theta_1},\ldots,e^{i \theta_M})\prod_{1 \leq j < k \leq M} |e^{i \theta_k}-e^{i \theta_j}|^4. 
\eeq
We recognize that, 
\beq
\mathbf{t}_{\lambda} = \frac{(2M)!}{2^M M! L^M}\mathbb{E}_{\beta=4}(|s_P|^4).
\eeq
 Thus the IPR of the ground state is equal to
\beql{ground}
\mathbf{t}_0=\frac{(2M)!}{2^M M! L^M}
\eeq
as determined previously in \cite{misguich2016inverse} for $L=2M$.

\section{Jack and Schur polynomials \label{schurjack}}

In this section we gather the results from symmetric function theory that we will use to analyze the IPRs. The primary sources are \cite{MR1676282} and \cite{MR3443860}.

\subsection{Partitions}

A partition $\lambda=[\lambda_1,\lambda_2,\ldots]$ is a weakly decreasing sequence of nonnegative integers. The length $l(\lambda)$ is the number of nonzero entries. Let $m_i(\lambda)=m_i$ be the number of times $i$ appears in $\lambda$. Set
\beql{zstff}
z_\lambda=\prod_{i \geq 1} i^{m_i} m_i!.
\eeq
The weight of a partition $|\lambda|$ is $\sum_i \lambda_i$. Given two partitions $\lambda,\mu$ of equal weight, the dominance partial order is defined by
\beq
\lambda \succeq \mu \iff \sum_{i=1}^k \lambda_i \geq \sum_{i=1}^k \mu_i \, \forall k.
\eeq 
The Young diagram of a partition $\lambda$ is a convenient graphical representation. It is obtained by putting $\lambda_i$ boxes left-aligned at row $i$. The conjugate partition $\lambda'$ of $\lambda$ is obtained by reflecting the Young diagram of $\lambda$ about the diagonal.

\begin{example} \underline{Young diagrams and conjugate partitions}. The Young diagram of $\lambda=[2,1,1]$ is
\[
\yng(2,1,1).
\]
The conjugate partition $\lambda'$ is read off from the Young diagram after reflection in the diagonal:
\[
\yng(3,1).
\]
Thus $\lambda'=[3,1]$.
\end{example}

\subsection{Jack polynomials}

The Jack polynomials $J_{\lambda}^{(2/\beta)}$ are a notable family of symmetric polynomials parameterized by a real parameter $\beta$. For $\beta=2$, they specialize to scaled Schur functions. For $\beta=1$ one obtains the zonal polynomials and for $\beta=4$ the quaternion zonal polynomials \cite{MR2325917}. We will be using the ``J'' normalization. The Jack polynomials will be useful for us due to their property of orthogonalizing the circular ensembles 
\cite[(10.36)]{MR3443860}. We state this more precisely now.

 Let $\lambda$ and $\kappa$ be a pair of partitions. Let $\lambda'$ denote the transpose of the partition $\lambda$ and $l(\lambda)$ the number of nonzero parts in $\lambda$. Let $(i,j) \in \lambda$ refer to a cell of the Young diagram of $\lambda$ (the indexing begins at (1,1)). Set 
\beq
\mathcal{N}_\lambda^{(\alpha)}(M)=\prod_{(i,j)\in \lambda} \frac{M+(j-1)\alpha-(i-1)}{M+j \alpha-i}
\eeq
and
\beq
C_\lambda(\alpha)=\prod_{(i,j) \in \lambda} (\alpha(\lambda_i-j)+\lambda_j'-i+1)(\alpha(\lambda_i-j)+\lambda_j'-i+\alpha).
\eeq

The Jack polynomials satisfy \cite{MR3433582}
\begin{multline}
\int_{[0,2 \pi]^{M}} \frac{d \theta_1}{2\pi} \cdots \frac{d \theta_M}{2\pi} J_\kappa^{(2/\beta)}(e^{i \theta_1},\ldots,e^{i \theta_n}) \overline{J_\lambda^{(2/\beta)} (e^{i \theta_1},\ldots,e^{i \theta_n})} \prod_{j < k}|e^{i \theta_j}-e^{i \theta_k}|^{\beta}
\\
=\delta_{\kappa,\lambda} \delta(l(\lambda) \leq M) C_\lambda(2/\beta) \mathcal{N}_\lambda^{(2/\beta)}(M) \frac{\Gamma(1+M \beta/2)}{\Gamma(1+\beta/2)^M}.
\end{multline}
The factor $\frac{\Gamma(1+M \beta/2)}{\Gamma(1+\beta/2)^M}$ is the one from Dyson's ex-conjecture \cite{dyson1962statistical},
\beq
\int_{[0,2 \pi]^{M}} \frac{d \theta_1}{2\pi} \cdots \frac{d \theta_M}{2\pi}  \prod_{j < k}|e^{i \theta_j}-e^{i \theta_k}|^{\beta}=\frac{\Gamma(1+M \beta/2)}{\Gamma(1+\beta/2)^M}.
\eeq
The factor $\delta(l(\lambda) \leq M)$ is simply the statement that $J_\lambda^{(2/\beta)}=0$ if the number of parts of $\lambda$ is greater than $M$. 

We define a corresponding inner product 
\beq
 \langle f,g \rangle_\beta=\frac{1}{M!} \int_{[0,2\pi]^M} \frac{d \theta_1 \cdots d \theta_M}{(2 \pi)^M} f \bar{g} \prod_{j < k}|e^{i \theta_j}-e^{i \theta_k}|^{\beta}
\eeq
with respect to which
\beql{finiteHall}
\langle J_{\kappa}^{(2/\beta)},J_{\lambda}^{(2/\beta)} \rangle_\beta=\delta_{\kappa,\lambda} C_\lambda(2/\beta) \mathcal{N}_\lambda^{(2/\beta)}(M) \frac{\Gamma(1+M \beta/2)}{ M! \Gamma(1+\beta/2)^M}.
\eeq
This inner product can be thought of as a specialization of the Hall inner product to finitely many variables. It may also be written as the extraction of a constant term (see e.g., \cite[VI, 10.35]{MR3443860})
\beql{Hall}
\langle f, g \rangle_{\beta}=\frac{1}{M!} \int_{T^M} f(z) \overline{g(z)} \Delta(z;\beta)=\frac{1}{M!}\CT[f \bar{g} \Delta(z; \beta)]
\eeq
\beq 
   \Delta(z;\beta)=\prod_{i \neq j} (1-z_i z_j^{-1})^{\beta/2}, \qquad T=\{z:|z|=1\}. 
\eeq

A word about notation. Since our main interest will be with $\beta=4$, we will reserve the shorter $\mathcal{N}_\lambda, \mathcal{C}_\lambda$ for that case.

The Jack polynomials also form an orthogonal basis of the symmetric polynomials with respect to the Hall inner product \cite[(4.1),(4.3)]{MR3433582},

\beq
\langle J_\lambda^{(2/\beta)}, J_\mu^{(2/\beta)} \rangle_{\beta}'=\delta_{\lambda \mu} C_\lambda(2/\beta)
\eeq
defined on power sum polynomials by
\beq
\langle p_\lambda, p_\mu \rangle_{\beta}' = \delta_{\lambda \mu} (2/\beta)^{l(\lambda)} z_\lambda.
\eeq
Conceptually, we will think of this inner product as being obtained from $\langle \cdot , \cdot \rangle_\beta$,  upon approximating $\mathcal{N}_\lambda^{(\beta)}(M) \approx 1$, which occurs for large $M$. The $\Gamma$ factors are a matter of normalization.

\subsection{Schur polynomials}
Let $s_\lambda(x_1,\ldots,x_M)$ denote the Schur polynomial associated to the partition $\lambda=(\lambda_1,\ldots,\lambda_M)$:
\beq
s_\lambda(x_1,\ldots,x_M):=\frac{\det[(x_i^{\lambda_j+M-j})_{1 \leq i,j \leq M}]}{\det[(x_i^{M-j})_{1 \leq i,j \leq M}]}.
\eeq

The Schur polynomials are (rescaled) Jack polynomials $J^{(1)}$ and form a linear basis for the symmetric polynomials.  There is a combinatorial description for the product of two Schur functions. Namely, writing
\beql{LR}
s_\lambda s_\mu = \sum_{\nu} c_{\lambda,\mu}^\nu s_\nu,
\eeq
the Littlewood-Richardson rule states that $c_{\lambda,\mu}^\nu$ is equal to the number of Littlewood-Richardson tableaux of skew shape $\nu/\lambda$ and of weight $\mu$. The coefficients are known as the Littlewood-Richardson coefficients, and appear in many other mathematical contexts.

\subsection{Transition matrices}

When discussing transition matrices between $\Q$-bases of symmetric polynomials \cite[ I \S 6]{MR3443860}, we index rows and columns by partitions of a positive integer $n$, arranged in reverse lexicographical order (so that $[n]$ is first and $[1^n]$ is last). A matrix $(M_{\lambda \mu})$ is \emph{strictly upper triangular} if $M_{\lambda \mu}=0$ unless $\mu \preceq \lambda$ in the dominance order on partitions. The strictly upper triangular matrices form a group. Given two $\Q$-bases $(u_\lambda), (v_\lambda)$, we denote by $M(u,v)$ the matrix $(M_{\lambda \mu})$ of coefficients in the equations
\beq
u_\lambda = \sum_\mu M_{\lambda \mu} v_\mu.
\eeq
$M(u,v)$ is called the \emph{transition matrix} from the basis $(u_\lambda)$ to the basis $(v_\lambda)$. 

The following result makes appearance in experimental form in \cite{gomez} and is derived using physical arguments in \cite{1110.6720}.

\begin{Proposition}\label{transitionsJ}
The transition matrix $M(s,J^{(2/\beta)})$ is strictly upper triangular with respect to the dominance ordering of partitions.
\end{Proposition}

\begin{proof}
Let $(m_\lambda)$ denote the the monomial basis for the symmetric polynomials. The transition matrices $M(s,m)$ and $M(J^{(2/\beta)},m)$ are strictly upper triangular. Consequently,
\beq
M(s,J^{(2/\beta)})=M(s,m)M(m,J^{(2/\beta)})=M(s,m)M(J^{(2/\beta)},m)^{-1}
\eeq
is strictly upper triangular.
\end{proof}

Define the transition coefficients $d_\nu^\lambda(\beta)=d_\nu$ from Schur polynomials to Jack polynomials:

\beq
s_\lambda(x_1,\ldots,x_n)=\sum_{\nu \preceq \lambda} d_\nu^\lambda J_\nu^{(2/\beta)}(x_1,\ldots,x_n).
\eeq

\section{Calculating the IPRs of XX}

Set 
\beq
\Delta(x;\beta)=\prod_{i \neq j} (1-x_i x_j^{-1})^{\beta/2}.
\eeq
Then
\beq
\Delta(x;\beta)=\prod_{i<j}[(x_j-x_i)(x_j^{-1}-x_i^{-1})]^{\beta/2}.
\eeq
In particular,
\beql{vandtorus}
\Delta(e^{i \theta};\beta)=\prod_{k < l} |e^{i \theta_l}-e^{i \theta_k}|^{\beta}.
\eeq
Define $T_L$ to be a discrete torus:
\beq
T_L:=\{e^{i \theta} \in \C  :  \theta \in \frac{2\pi \Z}{L} \, \forall j\}.
\eeq
Define a scalar product on symmetric polynomials by
\beql{discreteHall}
\langle f, g \rangle_{L;\beta} := \frac{1}{L^M M!} \sum_{T_L^M} f(z) \overline{g}(z) \Delta(z; \beta).
\eeq
In terms of this ``discrete" Hall inner product, 
\beql{IPR3}
\mathbf{t}_\lambda = \frac{\langle s_\lambda^2, s_\lambda^2 \rangle_{L;4}}{L^{M}}.
\eeq
Note that 
\beq
\sum_{z_1,\ldots,z_M \in T_L} z_1^{a_1} z_2^{a_2} \cdots z_M^{a_M} = \begin{cases} L^M & \text{ if }a_1 \equiv a_2 \equiv \cdots \equiv a_M \equiv 0 \pmod{L} \\
0 & \text{ otherwise}.
\end{cases}
\eeq
Consequently, the inner product can be expressed as an extraction of coefficients:
\beql{coefficientsmodL}
\langle f, g \rangle_{L;\beta}=\frac{1}{M!}\sum_{i_1,\ldots,i_M \in \Z} [x_1^{L i_1} \cdots x_M^{L i_M}] f(x_1,\ldots,x_M) g(x_1^{-1},\ldots,x_M^{-1}) \Delta(x; \beta).
\eeq

As we show next, for fixed $f,g,M$ and all sufficiently large $L$,
\beq
\langle f, g \rangle_{L;\beta} = \frac{1}{M!}\CT[f(x) g(x^{-1}) \Delta(x;\beta)]=\frac{1}{M!} \int_T f(z) \overline{g(z)} \Delta(z;\beta), 
\eeq
the Hall inner product of \eqref{Hall}.

\begin{Proposition}\label{conditionsforHall}
Let $f,g$ be symmetric polynomials and suppose that $\beta/2 \in \N$. Let $\deg_i(f)$ denote the degree of $x_i$ in $f$. Set
\beq
p_1:=\max_i \deg_i(f), \quad q_1:=\max_i \deg_i(g). 
\eeq
\beq
p_2:=\min_i \deg_i(f), \quad q_2:=\min_i \deg_i(g). 
\eeq
Then 
\beq
|\max_i \deg_i (f \bar{g} \Delta(x;\beta))| \leq \max\{|p_1-q_2|,|p_2-q_1|\}+(\beta/2)(M-1).
\eeq
\end{Proposition}

\begin{proof}
Let $\delta=(0,1,\ldots,M-1)$. The product $f \bar{g} \Delta(x;\beta)$ is a sum of monomials of the following form: with $\sigma^{(i_1)},\sigma^{(i_2)}$ permutations, $x_1^{f_1} \cdots x_M^{f_M}$ monomials coming from $f$ and $x_1^{g_1} \cdots x_M^{g_M}$ monomials from $g$,
\beq
x_1^{f_1} \cdots x_M^{f_M} x_1^{-g_1} \cdots x_M^{-g_M} \prod_{i=1}^{\beta/2} (x^{\sigma^{(i_1)}\delta} x^{-\sigma^{(i_2)}\delta}).
\eeq
The degree of $x_i$ is equal to
\beq
f_i-g_i+\sum_{i=1}^{\beta/2} \sigma^{(i_1)}(\delta)-\sigma^{(i_2)}(\delta).
\eeq
It satisfies
\beq
|f_i-g_i+\sum_{i=1}^{\beta/2} \sigma^{(i_1)}(\delta)-\sigma^{(i_2)}(\delta)| \leq \max\{|p_1-q_2|,|p_2-q_1|\}+(\beta/2)(M-1).
\eeq
\end{proof}

\begin{example}
\underline{Proposition \ref{conditionsforHall} is tight}. 
Taking $M=4, \lambda=[2,1,1]$ and $\beta=2$, we have
\beq
p_1=q_1=2, \quad p_2=q_2=0.
\eeq
The proposition guarantees that the maximum degree of a variable in $f \bar{g} \Delta(x;2)$ does not exceed $M+1=5$. This means that when $L>5$, the proposition guarantees that the discrete Hall inner product will be equal to the Hall inner product. Computing on the monomial symmetric functions,

\beq
\langle m_\lambda, m_\lambda \rangle_{4;2}=16
\eeq 
\beq
\langle m_\lambda, m_\lambda \rangle_{5;2}=9
\eeq
\beq
\langle m_\lambda, m_\lambda \rangle_{L;2}=\frac{1}{M!}\CT[m_\lambda \overline{m}_\mu \Delta(x)]=\langle m_\lambda, m_\lambda \rangle_2=2 \quad  \forall L \geq 6.
\eeq

We see that the proposition is tight in the sense that no smaller $L$ would work.
\end{example}

Proposition \ref{conditionsforHall} clarifies when we can truncate the sum in \eqref{coefficientsmodL}. Namely, each $i$ should satisfy
\beq
Li \leq \max\{|p_1-q_2|,|p_2-q_1|\}+(\beta/2)(M-1).
\eeq

The case $\beta=4$ is of particular  interest to us. According to Proposition \ref{conditionsforHall}, if $M,L, \lambda,\mu$ satisfy $L-2M>\max\{\lambda_1,\mu_1\}-2$, then 
\beq
\langle J_\lambda^{(1/2)},J_\mu^{(1/2)} \rangle_{L;4}=\langle J_\lambda^{(1/2)},J_\mu^{(1/2)} \rangle_{4}.
\eeq
In particular, under these conditions these Jack polynomials are orthogonal with respect to the scalar product $\langle \cdot , \cdot \rangle_{L;4}$.

\begin{Theorem}
Suppose that $L$ and $M$ are given. Let $\lambda$ be a fixed partition with $\lambda_1 <(1/2)(L-2M+2)$. Let $s_\lambda, J_\mu^{(1/2)}$ be the Schur and Jack polynomials, respectively. Write
\beq
s_\lambda^2=\sum_{\nu} r_\lambda^\nu J_\nu^{(1/2)}.
\eeq
The IPR for partition $\lambda$ is equal to
\beq
\mathbf{t}_\lambda= \frac{(2M)!}{M!(2L)^M} \sum_\nu \mathcal{N}_\nu^{(1/2)}(M) \mathcal{C}_\nu(1/2)    (r_\lambda^\nu)^2.
\eeq 
\end{Theorem}

\begin{proof}
Under these conditions on $\lambda$,
\beq
\mathbf{t}_\lambda = \frac{\langle s_\lambda^2, s_\lambda^2 \rangle_{L;4}}{L^{M}}=\frac{\langle s_\lambda^2, s_\lambda^2 \rangle_{4}}{L^{M}}.
\eeq
We have
\beq
\langle s_\lambda^2, s_\lambda^2 \rangle_{4}=\langle \sum_{\nu} r_\lambda^\nu J_\nu^{(1/2)}, \sum_{\nu} r_\lambda^\nu J_\nu^{(1/2)} \rangle_4=\sum_{\nu} (r_{\lambda}^{\nu})^2 \langle J_\nu^{(1/2)},J_\nu^{(1/2)} \rangle_4.
\eeq
Applying \eqref{finiteHall} for $\beta=4$ yields the result.
\end{proof}

\begin{example}
The simplest excited state has $\lambda=[1]$. Computing,
\beq
s_{[1]}^2=s_{[1,1]}+s_{[2]}.
\eeq
and
\beq
s_{[1,1]}=(1/2)J_{[1,1]}, \quad s_{[2]}=(2/3)J_{[2]}-(1/6)J_{[1,1]},
\eeq
so that
\beq
s_{[1]}^2=(2/3)J_{[2]}+(1/3)J_{[1,1]}.
\eeq
The IPR is
\beq
\mathbf{t}_{[1]}=(1/L^M)\langle s_{[1]}^2,s_{[1]}^2 \rangle_L=(1/L^M)\langle (2/3)J_{[2]}+(1/3)J_{[1,1]},(2/3)J_{[2]}+(1/3)J_{[1,1]} \rangle_L.
\eeq

We also compute
\beq
C_{[1,1]}=3/2, \quad C_{[2]}=3/4.
\eeq
\beq
\mathcal{N}_{[1,1]}=\frac{M}{M-1/2}\frac{M-1}{M-3/2}, \quad \mathcal{N}_{[2]}=\frac{M}{M-1/2}\frac{M+1/2}{M}.
\eeq
Consequently,
\beq
L^M \mathbf{t}_{[1]}=(4/9)\langle J_{[2]}, J_{[2]} \rangle_L+(4/9)\langle J_{[2]},J_{[1,1]} \rangle_L+(1/9)\langle J_{[1,1]}, J_{[1,1]} \rangle_L.
\eeq
Assuming $L>2M+2$, we can replace the discrete Hall inner products with the Hall inner product. Then
\beq
\langle J_{[2]},J_{[2]} \rangle_L=  C_{[2]} \mathcal{N}_{[2]} \frac{(2M)!}{M! 2^M}, \quad   \langle J_{[1,1]},J_{[1,1]} \rangle_L =  C_{[1,1]} \mathcal{N}_{[1,1]} \frac{(2M)!}{M! 2^M}, \quad \langle J_{[1,1]},J_{[2]} \rangle_L = 0.
\eeq
This gives us the exact value of the IPR $\mathbf{t}_{[1]}$. Assuming further that $M$ is large, so that $\mathcal{N} \approx 1$, yields
\beq
\mathbf{t}_{[1]} \approx (1/2) \frac{(2M)!}{M!L^M 2^M}=(1/2) \mathbf{t}_{0}.
\eeq
\end{example}

\begin{example} \underline{Table of IPRs}.
\begin{table}[H]
\centering
$\begin{tabu}{|l|l|}\hline
  \lambda & \mathbf{t}_\lambda \\
  \hline
  [] & \mathbf{t}_0 \\ \hline
  [1] & (1/2)\mathbf{t}_0 \\ \hline
   [2] & (11/32)\mathbf{t}_0 \\ \hline
   [1,1] & (11/32)\mathbf{t}_0 \\ \hline
   [3] & (17/64)\mathbf{t}_0 \\ \hline
     [2,1] & (1/4)\mathbf{t}_0 \\ \hline
       [1,1,1] & (17/64)\mathbf{t}_0 \\ \hline
         [4] & (1787/8192)\mathbf{t}_0 \\ \hline
        [3,1] & (1451/8192)\mathbf{t}_0 \\ \hline
          [2,2] & (99/512)\mathbf{t}_0 \\ \hline
            [2,1,1] & (1451/8192)\mathbf{t}_0 \\ \hline
            [1,1,1,1] & (1787/8192)\mathbf{t}_0 \\ \hline
\end{tabu}$
\caption{Table of IPRs, assuming $L-2M>2 \lambda_1$ and $\mathcal{N} \approx 1$, the latter occurring when $M$ is large.}
\end{table}

\end{example}

The table suggests the following property:
\beq
\langle J_\lambda^{(1)},J_\lambda^{(1)} \rangle_{4}'=\langle J_{\lambda'}^{(1)},J_{\lambda'}^{(1)} \rangle_{4}'.
\eeq
Here, we recall, $J^{(1)}$ is a (rescaled) Schur polynomial and the inner product is the Hall inner product for parameter $\beta=4$ (or $\alpha=1/2$ in the symmetric polynomial literature). We emphasize that the inner product has a different parameter than the Jack polynomials being operated on. Moreover, generally even
\beq
\langle J_{\lambda}^{(1/2)},J_{\lambda}^{(1/2)} \rangle_{4}' \neq \langle J_{\lambda'}^{(1/2)},J_{\lambda'}^{(1/2)} \rangle_{4}',
\eeq
where the parameters are matching. Thus the choice of Schur polynomials is special among the Jack polynomials.

\begin{Theorem}
Let $\lambda, \mu$ be partitions and $s_\lambda$ the Schur polynomial. For any $\beta$,
\beq
\langle s_\lambda, s_\mu \rangle_{\beta}'=\langle s_{\lambda'}, s_{\mu'} \rangle_{\beta}'
\eeq
\end{Theorem}

\begin{proof}
Following \cite[VI \S 2]{MR3443860}, consider the symmetric polynomials with coefficients in $\Q(q,t)$ and define an inner product on the basis of power sum polynomials $p_\mu$:
\beq
z_\mu(q,t):=z_\mu \prod_{i=1}^{l(\mu)} \frac{1-q^{\mu_i}}{1-t^{\mu_i}}, \quad \langle p_\mu, p_\kappa \rangle_{q,t}':=\delta_{\mu,\kappa}z_\mu(q,t).
\eeq
This inner product is a q-analogue of the inner product $\langle \cdot , \cdot \rangle_{\beta}'$. Indeed, denote the limit $(q,t) \rightarrow (1,1)$ with $q=t^{2/\beta}$ by $(q,t) \xrightarrow[2/\beta]{} (1,1)$. Then 
\beq
 \lim_{(q,t) \xrightarrow[2/\beta]{} (1,1)} \langle \cdot , \cdot \rangle_{(q,t)}'= \langle \cdot , \cdot \rangle_{\beta}'.
\eeq
Let $\omega_{t,q}$ be the standard automorphism on symmetric functions with coefficients in $\Q(q,t)$: 
\beq
\omega_{q,t}(p_\lambda)=(-1)^{|\lambda|+l(\lambda)} p_\lambda \prod_{i=1}^{l(\lambda)} \frac{1-q^{\lambda_i}}{1-t^{\lambda_i}}
\eeq
Let 
\beq
\omega_\beta=\lim_{(q,t) \xrightarrow[2/\beta]{} (1,1)} \omega_{q,t}.
\eeq
 The automorphism $\omega_{q,t}$ satisfies
\beq
\omega^{-1}_{q,t}=\omega_{t,q}
\eeq
\beq
\langle \omega_{u,v} f, g \rangle_{q,t}'=\langle  f, \omega_{u,v} g \rangle_{q,t}'
\eeq
\beq
\langle \omega_{t,q} f, g \rangle_{q,t}'= \langle \omega_2 f, g \rangle_{2}'
\eeq
\beq
\langle \omega_2 f, \omega_2 g \rangle_2'=\langle  f,  g \rangle_2'.
\eeq
The Schur polynomials satisfy
\beq
\omega_2 s_\lambda = s_{\lambda'}.
\eeq
Putting this together,
\beq
\langle  s_\lambda, s_\mu \rangle_{\beta}'=\langle  \omega_\beta^{-1} s_\lambda, \omega_\beta s_\mu \rangle_{\beta}'=
\langle \omega_2 s_{\lambda}, \omega_\beta s_{\mu} \rangle_{2}'.
\eeq
By evaluating on the power sum basis, we can check that
\beq
\omega_{q,t} \omega_2 =  \omega_2 \omega_{q,t}.
\eeq
Consequently,
\beq
\langle \omega_2 s_{\lambda}, \omega_\beta s_{\mu} \rangle_{2}'= \langle \omega_2 \omega_2 s_{\lambda}, \omega_2 \omega_\beta s_{\mu} \rangle_{2}'= \langle \omega_2 \omega_2 s_{\lambda},\omega_\beta  \omega_2  s_{\mu} \rangle_{2}'=\langle \omega_2  s_{\lambda'},\omega_\beta   s_{\mu'} \rangle_{2}',
\eeq
which proves the result.
\end{proof}

\begin{Corollary}
Let $\lambda$ be a partition satisfying $\lambda_1<(1/2)(L-2M+2)$ with $M \gg |\lambda|$ (so that $\mathcal{N}_\lambda \approx 1$). Then the IPRs satisfy the duality relation
\beq
\mathbf{t}_\lambda=\mathbf{t}_{\lambda'}.
\eeq
\end{Corollary}

\begin{proof}
We have shown that under these hypotheses,
\beq
\mathbf{t}_\lambda=\frac{\langle s_\lambda^2, s_\lambda^2 \rangle_4}{L^{M}}.
\eeq
Moreover, since $\mathcal{N}_\lambda \approx 1$,
\beq
\frac{\langle s_\lambda^2, s_\lambda^2 \rangle_4}{\langle s_{\lambda'}^2, s_{\lambda'}^2 \rangle_4}=\frac{\langle s_\lambda^2, s_\lambda^2 \rangle_4'}{\langle s_{\lambda'}^2, s_{\lambda'}^2 \rangle_4'}.
\eeq
It will suffice to show that the latter ratio is equal to $1$. Write $s_\lambda^2$ and $s_{\lambda'}^2$ in the Schur basis:
\beq
s_\lambda^2=\sum_\nu c_\nu^{\lambda} s_\nu, \quad 
s_{\lambda'}^2=\sum_\nu c_\nu^{\lambda'} s_\nu.
\eeq
Since $\omega$ is an automorphism, for any partitions $\mu,\kappa$,
\beq
\omega(s_\mu s_\kappa)=s_{\mu'} s_{\kappa'}.
\eeq
In particular,
\beq
s_\lambda^2=\omega \omega(s_\lambda^2)=\omega(s_{\lambda'}^2) \implies s_{\lambda'}^2=\omega(s_\lambda^2).
\eeq
Comparing coefficients,
\beq
\sum_{\nu} c_{\nu}^{\lambda'} s_{\nu}=\sum_{\nu} c_{\nu}^{\lambda} s_{\nu'}=\sum_{\nu'} c_{\nu'}^\lambda s_{\nu} \implies c_{\nu}^{\lambda'}=c_{\nu'}^{\lambda}.
\eeq
Expanding
\beq
\langle s_\lambda^2, s_\lambda^2 \rangle_4'=\sum_{\nu_1,\nu_2} c_{\nu_1}^{\lambda} c_{\nu_2}^{\lambda} \langle s_{\nu_1}, s_{\nu_2} \rangle_4'=\sum_{\nu_1,\nu_2} c_{\nu_1'}^{\lambda} c_{\nu_2'}^{\lambda} \langle s_{\nu_1'}, s_{\nu_2'} \rangle_4'=\sum_{\nu_1,\nu_2} c_{\nu_1}^{\lambda'} c_{\nu_2}^{\lambda'} \langle s_{\nu_1'}, s_{\nu_2'} \rangle_4'
\eeq
\beq
=\sum_{\nu_1,\nu_2} c_{\nu_1}^{\lambda'} c_{\nu_2}^{\lambda'} \langle s_{\nu_1}, s_{\nu_2} \rangle_4'=\langle s_{\lambda'}^2, s_{\lambda'}^2 \rangle_4'.
\eeq

\end{proof}

We have shown that conjugate partitions have asymptotically (assuming $\mathcal{N} \approx 1$) equal IPRs. This is not the case for finite $M$, when $\mathcal{N} \neq 1$. It would be interesting to give a physical explanation for this duality.

\subsection{IPRs and diagonals of permutahedra}

An alternate expression for $\sum_{\mathbf{x}} |c(\mathbf{x})|^4$ is obtained by multiplying out the terms of the determinant:
\beq
\sum_{\mathbf{x}} |c(\mathbf{x})|^4=\frac{1}{M!} \sum_{P,Q,R,S \in S_M} \sgn(PQRS) \prod_{j=1}^M \sum_{x_1,\ldots,x_M=1}^L e^{i x_j(k_{P(j)}+k_{Q(j)}-k_{R(j)}-k_{S(j)})}=
\eeq
\beq
=\frac{L^M}{M!} \sum_{P,Q,R,S \in S_M} \sgn(PQRS) \delta(k_{P(j)}+k_{Q(j)}-k_{R(j)}-k_{S(j)} \equiv 0 \mod{2 \pi}, \forall j).
\eeq
Using $k_j = \frac{2 \pi I_j}{L}$, this sum is equal to
\beq
\frac{L^M}{M!} \sum_{P,Q,R,S \in S_M} \sgn(PQRS) \delta(I_{P(j)}+I_{Q(j)}-I_{R(j)}-I_{S(j)} \equiv 0 \mod{L}, \forall j)
\eeq
\beq
=\frac{L^M}{M!} \sum_{P,Q,R,S \in S_M} \sgn(PQRS) \delta(I_{P}+I_{Q}-I_{R}-I_{S} \equiv 0 \mod{L})
\eeq
Consequently,
\beql{count}
\mathbf{t}_{\lambda}=\frac{1}{L^M M!} \sum_{P,Q,R,S \in S_M} \sgn(PQRS)  \delta(I_{P}+I_{Q}-I_{R}-I_{S} \equiv 0 \mod{L}).
\eeq
 
We may interpret the sum as follows. A  permutahedron is the polytope obtained by taking the  convex hull of all permutations of a fixed vector $(I_1,I_2,\ldots,I_M)$ with distinct entries. We can interpret $I_{P}-I_{R}$ as a diagonal vector of this permutahedron. Then a condition of the form
\beq
I_P-I_R=I_S-I_Q
\eeq
would be a condition on two diagonals to be mutual translates and a congruence of the form
\beq
I_P-I_R \equiv I_S-I_Q \pmod{L}
\eeq
is interpreted similarly.

\section{Future Directions}

Several interesting questions remain to be studied. 

\bigskip

1. Experimental results suggest that 
\beq
\lim_{M \rightarrow \infty}  \frac{\langle J_\lambda^{(1/2)}, J_\lambda^{(1/2)} \rangle_{2M;4} }{\langle J_\lambda^{(1/2)}, J_\lambda^{(1/2)} \rangle_4 }=1.
\eeq
In other words, for a finite number of variables, the discrete Hall inner product is asymptotically equal to the Hall inner product when $L=2M$. The reason this is not a simple consequence of the convergence of a Riemann sum to the corresponding integral is that even though the number of sample points ($L$) on each torus increases, the number of torii ($M$) does as well. Nonetheless, experiment suggests that the error decreases (see Figure \ref{asymptotics}). If true, this will allow to extend the methods to compute the asymptotics of IPRs when $L=2M$.

\begin{figure}
\includegraphics[width=7in]{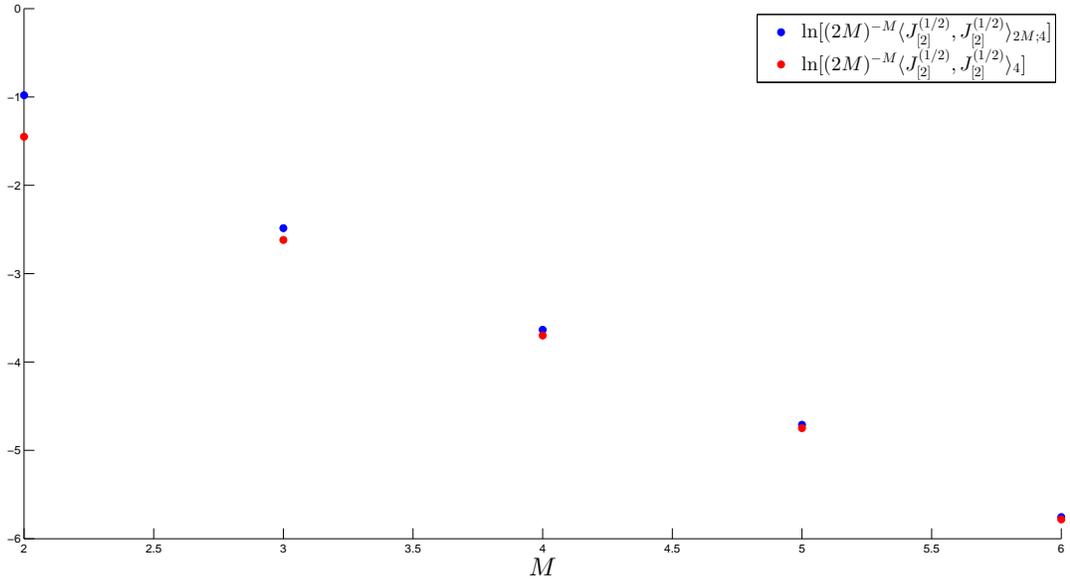}
\caption{\label{asymptotics} Plotted are the logs of $\frac{\langle J_{[2]}^{(1/2)},J_{[2]}^{(1/2)} \rangle_{L;4}}{L^M}$ and $\frac{\langle J_{[2]}^{(1/2)},J_{[2]}^{(1/2)} \rangle_{4}}{L^M}$ for $L=2M$.}
\end{figure}

\bigskip

2. There exists a dynamical interpretation of the IPRs (described in \cite{misguich2016inverse}). In the case when the spectrum is nondegenerate, it is given by summing the IPRs:
\beq
T:=\sum_k t_k=\sum_{i,k=1}^D | \langle a_i | \psi_k \rangle|^4
\eeq
This quantity measures how much the eigenstates are localized in the preferential basis. It can range from $T_{\text{min}}=1$ (the eigenstates are spread maximally over the whole basis) to $T_{\text{max}}=D$ (each eigenstate matches a basis vector). The ratio $T/D$ measures the stationary return probability to an initial basis state, averaged over all the basis states. The minimum value is reached if the dynamics connects any initial basis state to all the other basis states. On the other hand, if it takes on the maximum value, then the system does not evolve if it is initialized from a basis state.

\bigskip

3. The scalar product $\langle \cdot , \cdot \rangle_{L;\beta}$ is a discretization of the Hall inner product. Using the Gram-Schmidt orthonormalization procedure on the symmetric monomial functions, we may define an orthonormal basis which is triangular with respect to the monomial symmetric basis. Such ``discrete Jack polynomials" may have interesting properties.

\bigskip

4. The data table suggests that the IPR of the ground state is largest. The next few excited states, labeling by partitions, have IPRs which are approximately $1/2,1/3,1/4$ and $1/5$ of the ground state IPR, for $|\lambda|=1,2,3,4$. It is unlikely that this pattern will continue. It would be interesting to determine the distribution of the values of the IPRs.

\bigskip
{\bf Acknowledgments}.  
The author would like to thank Vir Bulchandani, Mark Haiman, Joel Moore and Lauren Williams for their helpful comments in the writing of this manuscript.

\bibliographystyle{alpha}
\bibliography{bibliography}

\end{document}